\documentclass{iacrcc}
\license{CC-by}

\usepackage{graphicx}
\usepackage{booktabs}

\newcommand{\Cl}{\mathcal{C}\ell}
\newcommand{\OO}{\mathcal{O}}
\newcommand{\Zq}[1]{\mathbb{Z}_{#1}}
\newcommand{\act}{\star}

\title[running={CRT-Decomposed $\Sigma$-Protocols for CSIDH},
       plaintext={CRT-Decomposed Sigma-Protocols for CSIDH Class Group Actions}]
{CRT-Decomposed $\Sigma$-Protocols for CSIDH Class Group Actions}

\genericfootnote{This paper is under review in IACR Communications in Cryptology}
\addauthor[inst={1}]{I. Dey}
\addauthor[inst={1}]{I. Cherkaoui}
\addaffiliation[country={Ireland}]{South East Technological University, Waterford}

\begin{document}
\maketitle


\begin{abstract}
We construct a zero-knowledge proof of knowledge for the CSIDH group action
that exploits the Chinese Remainder Theorem (CRT) structure of the ideal
class group, available whenever the group structure is exactly known, as it
is for CSIDH-512. The protocol has three proven properties: perfect
completeness, perfect special honest-verifier zero-knowledge, and 2-special
soundness in which the secret is recovered from two accepting transcripts by
one subtraction and one modular inversion per CRT component. No rewinding
loop, lattice reduction, or heuristic sampling appears in the extractor.
Because each round admits exactly two responses, Unruh's transform yields a
non-interactive proof with straight-line extraction in the quantum random
oracle model (QROM), which removes the multiplicative forking-lemma loss. We
instantiate the scheme on the exactly known CSIDH-512 class group, verify
the complete algebraic layer by machine over the exact 258-bit modulus
($10^4$ random instances, all passing), and report protocol-level
simulation results in the exponent model: Monte Carlo soundness rates
matching the proven $2^{-t}$ bound within 95\% confidence at every tested
$t$, serialized signature sizes within 1.6\% of the size formulas,
instrumented action counts, and a scaled meet-in-the-middle attack whose
measured cost follows the predicted $\sqrt{q}$ law. We further prove two
delimiting results: CRT decomposition cannot enlarge the per-round
challenge space, and publishing the CRT hop curves lowers classical
key-recovery cost from about $2^{128.6}$ to about $2^{67.3}$ group action
evaluations. The construction is therefore correct and structurally
complete today, and quantitatively secure only on future parameter sets
whose class number has large prime factors. We compare against CSI-FiSh,
CSI-Otter, and Tanuki; the tight security of a blind signature built on
this proof of knowledge is left to future work.
\end{abstract}

\begin{textabstract}
We construct a zero-knowledge proof of knowledge for the CSIDH group
action that exploits the Chinese Remainder Theorem structure of the ideal
class group, available whenever the group structure is exactly known, as
it is for CSIDH-512. The protocol has perfect completeness, perfect
special honest-verifier zero-knowledge, and 2-special soundness in which
the secret is recovered from two accepting transcripts by one subtraction
and one modular inversion per CRT component, with no rewinding loop,
lattice reduction, or heuristic sampling. Because each round admits
exactly two responses, Unruh's transform yields a non-interactive proof
with straight-line extraction in the quantum random oracle model,
removing the multiplicative forking-lemma loss. We instantiate the scheme
on the exactly known CSIDH-512 class group, machine-verify the complete
algebraic layer over the exact 258-bit modulus, and report protocol-level
simulation results: Monte Carlo soundness rates matching the proven
bound, serialized signature sizes within 1.6 percent of the size
formulas, instrumented action counts, and a scaled meet-in-the-middle
attack following the predicted square-root law. We further prove two
delimiting results: CRT decomposition cannot enlarge the per-round
challenge space, and publishing the CRT hop curves lowers classical
key-recovery cost from about 2^128.6 to about 2^67.3 group action
evaluations. The construction is correct and structurally complete today,
and quantitatively secure only on future parameter sets whose class
number has large prime factors. Comparisons with CSI-FiSh, CSI-Otter, and
Tanuki are given; the tight security of a blind signature built on this
proof of knowledge is left to future work.
\end{textabstract}

\section{Introduction}
A blind signature lets a user obtain a signature on a
message without the signer learning the message. This is the mechanism
behind private digital cash, anonymous credentials, and unlinkable payment
authorization: a bank certifies a coin without seeing where the coin will
be spent. Deployed instantiations rest on discrete logarithms or RSA, which a
large quantum computer breaks; traffic recorded today can be attacked
once such a machine exists, so payment infrastructure with decade-long
lifetimes needs post-quantum replacements now. Isogeny-based group
actions are among the few post-quantum foundations supporting the
Schnorr-like structure blind signatures are built from. The leading
isogeny-based candidates, CSI-Otter~\cite{csiotter} and the Tanuki
frameworks~\cite{tanuki}, are built from interactive identification
protocols. Their security proofs share a weak point: to argue that a forger
must know the secret key, the proof rewinds the forger and applies the
forking lemma~\cite{hkl}, which loses at least a quadratic factor and is
delicate to carry out against quantum adversaries.

When a security proof turns any attacker into a solver
for a hard problem; a lossy transformation forces larger, slower
parameters. Our goal is to remove the lossiest step, rewinding extraction,
using algebraic structure instead. CSIDH~\cite{csidh} is built from a
commutative group, the ideal class group $\Cl(\OO)$, acting on a set of
elliptic curves. For CSIDH-512, and so far only for CSIDH-512, this group
was computed exactly by Beullens, Kleinjung, and Vercauteren~\cite{csifish}:
it is cyclic of a known composite order $N$ whose prime factorization is
known. Any cyclic group of known composite order splits by the Chinese
Remainder Theorem (CRT) into independent components, one per prime factor.
We design the proof of knowledge so that each component is proved, and can
be extracted, independently. Extraction then becomes a short exact formula
rather than a search.

The primary contribution of this paper is five-fold. First, we formalize weak
extractability for cryptographic group actions (Def.~\ref{def:wext}). This gives a precise meaning to a security idea: anyone who can convince a verifier must actually hold the secret key, not merely appear to. Without a formal definition, that idea can be neither proven nor disproven. Second, we construct the protocol $\Pi_{\mathrm{CRT}}$ and prove four
properties: perfect completeness, perfect special honest-verifier
zero-knowledge, 2-special soundness with a purely algebraic extractor, and
exact knowledge error $2^{-t}$ (Thms.~\ref{thm:comp}--\ref{thm:ke}). The
extractor recovers the secret by one subtraction and one modular inversion
per component. No rewinding loop appears. This builds the identification scheme itself and guarantees three things: honest users always succeed, an eavesdropper learns nothing from watching, and the secret can be computed from two correct answers by simple arithmetic, the way a speed follows from two position readings. Comparable schemes instead prove this by re-running the adversary from a saved state, which weakens the guarantee and fails against quantum attackers. Third, we instantiate the
protocol on CSIDH-512. The full algebraic layer is machine-verified over
the exact class number, and protocol-level simulations in the exponent
model confirm the soundness bound, the size formulas, and the attack
scaling law (Sec.~\ref{sec:concrete}). This shows the scheme is not a paper design. Every parameter is fixed for a real deployment target, a computer has checked the mathematics on ten thousand random examples, and simulations confirm the statistical behavior, the exact signature size in bytes, and the predicted cost of the best known attack.

Fourth, we prove two delimiting
results. CRT decomposition cannot enlarge the per-round challenge space
(Prop.~\ref{prop:obstruction}), and publishing the hop curves reduces
classical key recovery from $2^{128.6}$ to $2^{67.3}$ action evaluations
(Prop.~\ref{prop:decomp}). The construction's limits are therefore
quantified, not conjectured. Here we prove the technique cannot make the protocol shorter, and we calculate exactly how much security the design gives up through its published intermediate values, so a practitioner knows the price before paying it. Fifth, we compile the protocol to the QROM
through published theorems (Prop.~\ref{prop:qrom}). The result is a
non-interactive proof of knowledge with straight-line extraction and no
forking-lemma loss. This removes the interaction entirely, yielding a signature whose security argument survives quantum adversaries and does not degrade through the rewinding step that is the standard weak point of such proofs.

The paper is organized as follows. Section~\ref{sec:prelim}
presents notation, explains the group action in elementary terms, and states
the two definitions everything else depends on: known-structure instances
and weak extractability. Section~\ref{sec:protocol} specifies the protocol and proves its four properties, with a physical
interpretation after each theorem explaining what the statement means
outside the formalism. Section~\ref{sec:limits} proves what the
construction cannot do, so that the reader knows the boundary of the claims
before seeing the numbers. Section~\ref{sec:concrete} instantiates the
protocol on CSIDH-512 and presents the three evidence layers: machine
verification of the algebra, the security-decomposition bound with its
consequences, and the protocol-level simulations. Section~\ref{sec:sota}
positions the scheme against published isogeny-based schemes, including the
costs, and Section~\ref{sec:qrom} lifts the interactive protocol to a
non-interactive one in the quantum random oracle model, while
Section~\ref{sec:conc} presents the conclusion.

\section{Preliminaries}\label{sec:prelim}
We first recall the CSIDH group action, then define the two notions the rest of the paper depends on: a
\emph{known-structure instance}, which is the setting in which the CRT
decomposition is available at all, and \emph{weak extractability}, which
makes precise the property our protocol is designed to achieve. 

Let $p$ be a CSIDH prime and $\mathcal{E}$ the set of supersingular
elliptic curves over $\mathbb{F}_p$ with $\mathbb{F}_p$-endomorphism ring
$\OO$. The class group $\Cl(\OO)$ acts freely and transitively on
$\mathcal{E}$~\cite{csidh}; we write the action as
$\act:\Cl(\OO)\times\mathcal{E}\to\mathcal{E}$. Freeness and transitivity
mean that for every ordered pair of curves there is exactly one group element
mapping the first to the second. In this case, curves are positions and group elements are
displacements; applying a displacement is easy, recovering the
displacement between two positions is believed hard (vectorization, or
the group action inverse problem (GAIP)).

\begin{definition}[Known-structure instance]\label{def:ks}
A known-structure CSIDH instance is a tuple
$(p,\OO,g,N,(q_1,\dots,q_k))$ where $\Cl(\OO)=\langle[g]\rangle$ is cyclic
of order $N=\prod_{i=1}^{k}q_i$ with pairwise coprime prime powers $q_i$,
and where $[g^a]\act E$ can be evaluated efficiently for every
$a\in\Zq{N}$, as realized for CSIDH-512 by the lattice-based evaluation
of~\cite{csifish}.
\end{definition}

For $i\in[k]$ set $u_i:=N/q_i$, $v_i:=u_i^{-1}\bmod q_i$, and the CRT
idempotents $\varepsilon_i:=u_iv_i\bmod N$. Every $s\in\Zq{N}$ satisfies
$s\equiv\sum_i s_i\varepsilon_i\pmod N$ with $s_i:=s\bmod q_i$. Let
$h_i:=g^{u_i}$, an element of order $q_i$. Then
\begin{equation}\label{eq:crt}
[g^{s}] \;=\; \prod_{i=1}^{k}\,[h_i^{\,s_iv_i}] .
\end{equation}

\begin{definition}[Weak extractability]\label{def:wext}
A proof system for
$R=\{((E_0,E_1),s): E_1=[g^{s}]\act E_0\}$
is weakly extractable with error $\kappa$ if an extractor, given any prover
that convinces the verifier with probability $\epsilon>\kappa$, outputs $s$
with probability $\mathrm{poly}(\epsilon-\kappa)$ in expected polynomial
time. It is algebraically weakly extractable if the extractor's computation
on accepting transcripts uses only ring operations in $\Zq{N}$: no lattice
reduction, no sampling heuristics, no discrete logarithm computation.
\end{definition}

\section{The Protocol \texorpdfstring{$\Pi_{\mathrm{CRT}}$}{Pi-CRT} and Its Proven Properties}\label{sec:protocol}
We first describe key generation, which decomposes the secret along the CRT
components and publishes a chain of intermediate-hop curves; then the
single-round interaction; and then the four theorems. The order of the
theorems follows the order in which a verifier would care about them:
honest executions succeed (completeness), transcripts leak nothing
(zero-knowledge), successful provers must know the secret and the secret
can be read off algebraically (special soundness), and cheating survives
repetition only with exponentially small probability (knowledge error).
Each theorem is followed by a short interpretation of what it means
physically.

\subsubsection*{Key generation}
On input a known-structure instance and a base curve $F_0:=E_0$, sample
$s\leftarrow\Zq{N}$, set $\sigma_i:=s_iv_i\bmod q_i$, and define the hop
curves
\begin{equation}\label{eq:chain}
F_i \;:=\; [h_i^{\,\sigma_i}]\act F_{i-1},\qquad i=1,\dots,k .
\end{equation}
By~\eqref{eq:crt} the chain ends at $F_k=[g^{s}]\act E_0=:E_1$. The public
key is $(F_0,\dots,F_k)$; the witness is $s$, equivalently
$(\sigma_1,\dots,\sigma_k)$. The prover proves each hop relation
$R_i=\{((F_{i-1},F_i),\sigma_i): F_i=[h_i^{\sigma_i}]\act F_{i-1}\}$.

\subsubsection*{One round of hop $i$}
The prover samples $r\leftarrow\Zq{q_i}$ and sends
$T:=[h_i^{\,r}]\act F_{i-1}$. The verifier sends a challenge bit
$c\in\{0,1\}$. The prover replies $z:=r-c\,\sigma_i\bmod q_i$. The verifier
accepts if and only if
\begin{equation}\label{eq:verify}
T \;=\; [h_i^{\,z}]\act F_{i-1+c}\, .
\end{equation}
The full protocol runs $t$ parallel rounds for each of the $k$ hops.

\begin{theorem}[Perfect completeness]\label{thm:comp}
An honest prover is always accepted.
\end{theorem}
\begin{proof}
For $c=0$, equation~\eqref{eq:verify} is the definition of $T$. For $c=1$,
$[h_i^{\,r-\sigma_i}]\act F_i
=[h_i^{\,r-\sigma_i}]\act([h_i^{\,\sigma_i}]\act F_{i-1})
=[h_i^{\,r}]\act F_{i-1}=T$.
\end{proof}

\begin{theorem}[Perfect special HVZK]\label{thm:zk}
For each hop there is a simulator that, given the hop statement and a
challenge $c$, outputs a transcript with exactly the distribution of a real
one, without knowing the secret.
\end{theorem}
\begin{proof}
Sample $z\leftarrow\Zq{q_i}$ and set $T:=[h_i^{\,z}]\act F_{i-1+c}$. In a
real transcript $r$ is uniform, so $z=r-c\sigma_i$ is uniform on
$\Zq{q_i}$, and $T$ is the unique curve determined by $(z,c)$
through~\eqref{eq:verify}. Both distributions are uniform in $z$ with $T$ a
fixed function of it, hence identical.
\end{proof}

A transcript is a permanent record. Theorem~\ref{thm:zk} says this record carries exactly zero information
about the secret, not merely computationally hidden information: even an
unbounded observer learns nothing, because an identical record can be
manufactured without the secret. The key's long-term confidentiality thus
rests on no assumption at the transcript level; only the public key
exposes it to attack.

\begin{theorem}[2-special soundness, algebraic extraction]\label{thm:ss}
Two accepting transcripts $(T,0,z_0)$ and $(T,1,z_1)$ for hop $i$ determine
the witness as $\sigma_i=z_0-z_1\bmod q_i$, and the global secret is
\begin{equation}\label{eq:extract}
s \;=\; \sum_{i=1}^{k} \big((z_0^{(i)}-z_1^{(i)})\,v_i^{-1}\bmod q_i\big)\,
\varepsilon_i \;\bmod N .
\end{equation}
Hence $\Pi_{\mathrm{CRT}}$ is algebraically weakly extractable.
\end{theorem}
\begin{proof}
Acceptance gives $[h_i^{\,z_0}]\act F_{i-1}=T=[h_i^{\,z_1}]\act F_i$, so
$F_i=[h_i^{\,z_0-z_1}]\act F_{i-1}$. The action is free, so the exponent is
unique modulo $q_i$ and $\sigma_i=z_0-z_1$. Since $\gcd(v_i,q_i)=1$,
$s_i=\sigma_iv_i^{-1}\bmod q_i$ is well defined, and~\eqref{eq:extract} is
CRT recombination. Every step is a ring operation in $\Zq{N}$.
\end{proof}

Equation~\eqref{eq:extract} says the
secret is literally the difference of two answers to the same question
under the two challenges: knowledge of the key is localized in the
transcripts, the way a displacement is determined by two position
readings. This later converts into straight-line extraction, since a
reduction can read the secret off the adversary's own messages instead of
rerunning it from a saved state, the step that fails against quantum
attackers whose internal state cannot be copied.

\begin{theorem}[Knowledge error]\label{thm:ke}
With $t$ parallel rounds per hop, $\Pi_{\mathrm{CRT}}$ is a proof of
knowledge for $\bigwedge_i R_i$ with knowledge error $2^{-t}$.
\end{theorem}
\begin{proof}[Proof sketch]
The verifier accepts only if all $kt$ rounds accept. A prover that can
answer both challenges of some round of hop $i$ yields $\sigma_i$ by
Thm.~\ref{thm:ss}. A prover that can answer at most one challenge per round
of hop $i$ passes that hop with probability at most $2^{-t}$, and failing
one hop fails the whole proof. Any prover with success probability above
$2^{-t}$ therefore admits, for each hop, a round in which both challenges
are answerable with noticeable probability, and the standard per-hop
rewinding extractor runs in expected time
$\mathrm{poly}/(\epsilon-2^{-t})$.
\end{proof}

Each round is an independent
challenge that a prover without the key answers correctly with probability
at most one half, like a coin flip that must land favorably. After
$t=128$ rounds the probability of surviving by luck is below $2^{-128}$,
smaller than the probability of guessing a 128-bit key outright, which is
the standard bar for cryptographic soundness.

\begin{remark}[What the CRT structure buys]
The benefit is structural, not statistical. Extraction is component-wise
and closed-form~\eqref{eq:extract}; components can be proved, extracted,
delegated, or aggregated independently, which is the property needed for
distributed provers and for the online extraction of
Sec.~\ref{sec:qrom}. The benefit is not compactness and not a smaller
number of rounds, as the next section makes precise.
\end{remark}

\section{Limits of the Approach}\label{sec:limits}
Before presenting numbers, this section proves two negative results that
delimit what the CRT structure can and cannot deliver, so that the
positive claims of the paper are read at their correct strength. The first
shows that the decomposition cannot shorten the protocol. The second,
developed quantitatively in Section~\ref{sec:concrete}, identifies the
information cost of publishing the hop curves.

\begin{proposition}[No challenge-space amplification]\label{prop:obstruction}
Fix a round of hop $i$ in which the verifier's check has the form
$T=[h_i^{\,z}]\act X_c$ for a curve $X_c$ the verifier possesses, and in
which verification performs $O(1)$ action evaluations. Then the challenge
space is at most as large as the set of published candidate curves $X_c$.
Replacing binary challenges by challenges in $\Zq{m}$ requires publishing
$m$ curves per hop or performing $\Theta(m)$ action evaluations at
verification. CRT decomposition does not change this.
\end{proposition}
\begin{proof}
Correctness for challenge $c$ and response $z=r-c\sigma_i$ forces
$X_c=[h_i^{\,c\sigma_i}]\act F_{i-1}$. A group action, unlike
exponentiation in a group, gives the verifier no way to compute this curve
from $F_{i-1}$, $F_i$, and $c$ other than applying the action $c$ times or
receiving the curve directly.
\end{proof}

\subsubsection*{Result 3: protocol-level simulation}
We implemented the complete scheme in this model,
including the Unruh compilation with a SHA-256 random oracle and byte
serialization, and ran four simulations. (S1) \emph{Soundness}: the
extremal cheating prover of Thm.~\ref{thm:ke}, which knows every
component but one, was run for $4\times10^{5}$ independent protocol
executions at each $t=1,\dots,12$; the measured success rate matches the
proven $2^{-t}$ bound within the Wilson 95\% confidence interval at every
$t$ (points with error bars in Fig.~\ref{fig:soundness}). (S2)
\emph{Sizes}: serialized signatures measure within $1.6\%$ of the size
formulas at all tested $(t,k)$, the small excess coming from rounding
each response to whole bytes (triangular markers in
Fig.~\ref{fig:sizes}). (S3) \emph{Counts}: instrumented action counts
equal $tk$ for signing and for verification exactly, grounding the
projections of Fig.~\ref{fig:time} in measured counts rather than
formulas alone. (S4) \emph{Attack}: the meet-in-the-middle adversary of
Prop.~\ref{prop:decomp} was implemented in toy subgroups of order
$2^{10}$ to $2^{22}$, thirty runs each; the measured mean cost follows
the predicted $1.5\sqrt{q}$ evaluations (table plus expected half probe)
and never exceeds the $2\sqrt{q}$ worst case
(Fig.~\ref{fig:attack}). We state the scope plainly: these simulations
validate the protocol logic, the exact statistical behavior, the byte
formats, and the attack scaling law; by construction they measure neither
isogeny running time nor cryptographic hardness, and all wall-clock
figures elsewhere remain projections from the cited 40 ms/action
measurement.

\begin{figure}[t]
\centering
\includegraphics[width=0.66\textwidth]{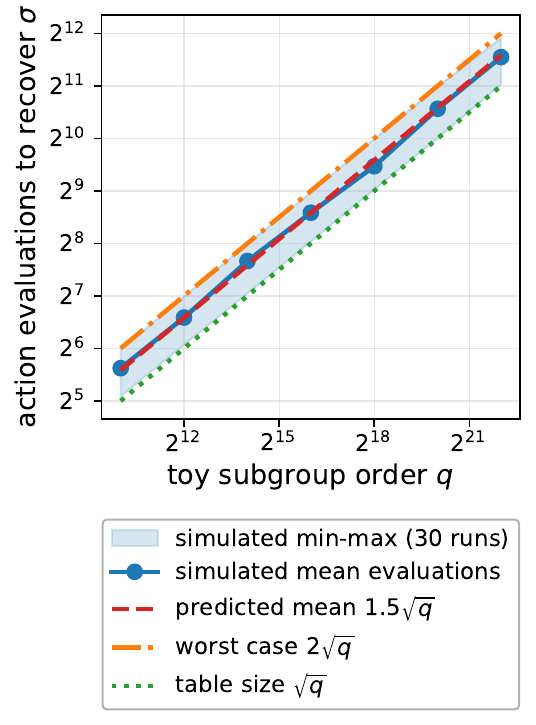}
\caption{Simulated meet-in-the-middle attack of Prop.~\ref{prop:decomp}
in toy subgroups (exponent model): measured mean action evaluations over
30 runs with min-max band, versus the predicted mean $1.5\sqrt{q}$, the
worst case $2\sqrt{q}$, and the $\sqrt{q}$ table size. Provenance:
measurement of this paper's released simulator. The simulation validates
the $\sqrt{q}$ scaling law of the security-decomposition bound; it makes
no claim about isogeny running time.}
\label{fig:attack}
\end{figure}

\begin{figure}[t]
\centering
\includegraphics[width=0.72\textwidth]{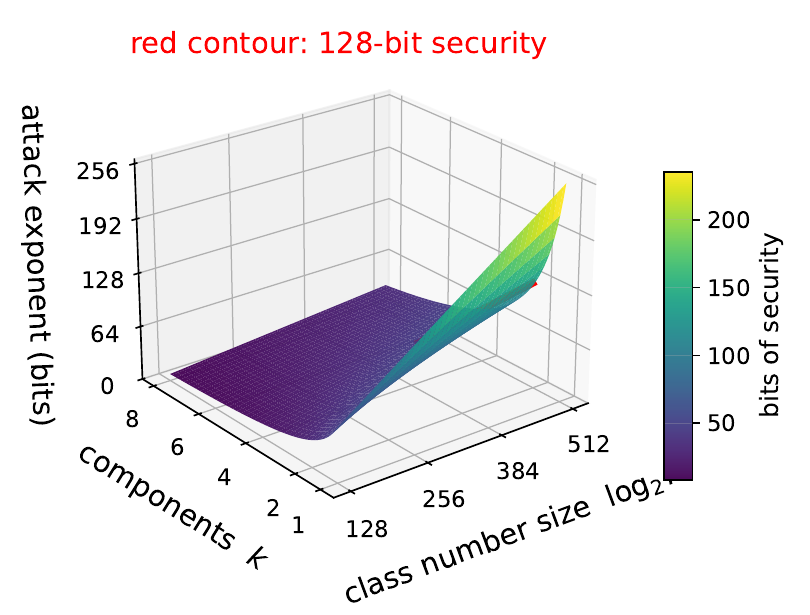}
\caption{Classical meet-in-the-middle attack exponent as a function of
class number size $\log_2 N$ and number of equal-size CRT components $k$,
$z=\log_2 N/(2k)$. Provenance: exact formula from Prop.~\ref{prop:decomp};
unequal splits such as CSIDH-512 are governed by the largest component.
The red contour marks 128-bit security: $k=2$ needs $\log_2 N\ge512$ and
$k=5$ needs $\log_2 N\ge1280$.}
\label{fig:surface}
\end{figure}

\section{Comparison with the State of the Art}\label{sec:sota}
This section answers the question a practitioner asks first: what does
this scheme cost relative to what exists, and what does the cost buy? The
table gives the published sizes of prior isogeny-based schemes next to the
parameter-determined sizes of ours, and the figures show the trade-offs
across the round parameter. The section is deliberately written so that
the disadvantages appear with the same prominence as the advantages.

Table~\ref{tab:sota} compares the instantiated scheme with published
isogeny-based schemes at the CSIDH-512 parameter level. The published
figures for prior schemes come from their papers~\cite{csifish,csiotter};
figures for this work are parameter arithmetic from Sec.~\ref{sec:concrete}
(no isogeny implementation exists yet, and we present no measured signing
times of our own). Figures~\ref{fig:sizes} and~\ref{fig:time} plot the
underlying trade-offs across the round parameter $t$.

\begin{figure}[t]
\centering
\includegraphics[width=0.8\textwidth]{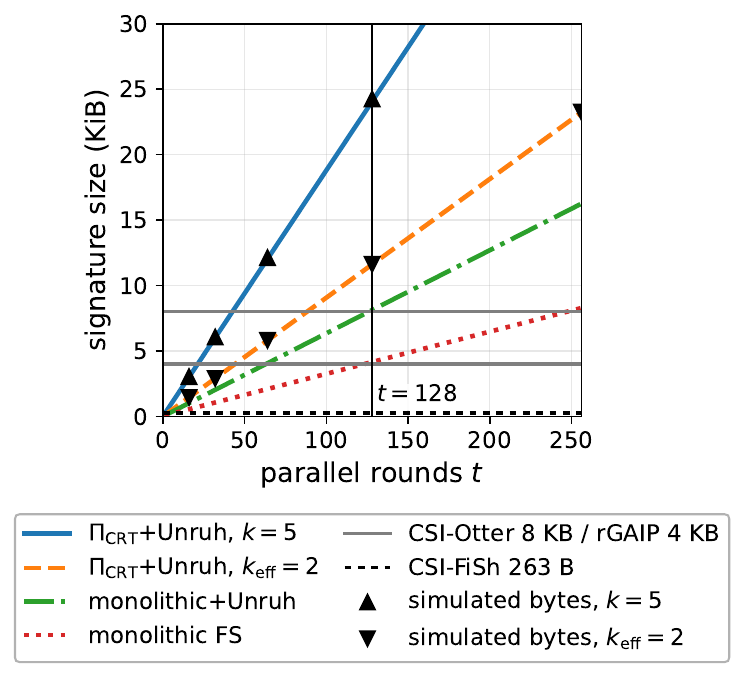}
\caption{Signature size versus rounds $t$: this work with $k=5$ and
$k_{\mathrm{eff}}=2$ under the Unruh transform, the monolithic
known-structure protocol with Unruh, and monolithic plain Fiat-Shamir.
Horizontal reference lines mark the published CSI-Otter (8 KB and 4 KB)
and CSI-FiSh (263 B) sizes. Provenance: exact size formulas and published values; triangular
markers are byte lengths of signatures serialized by the released
simulator, within 1.6\% of the formulas (byte-alignment of responses).
Online QROM extraction costs roughly 3 to 6 times CSI-Otter's size at
equal round count.}
\label{fig:sizes}
\end{figure}

\begin{table}[t]
\caption{Comparison at the CSIDH-512 level. Sizes for prior schemes are the
published values; sizes for this work are exact parameter arithmetic;
projected times use the published 40 ms per group action~\cite{csifish}.
PoK: proof of knowledge. OMUF: one-more unforgeability.}
\label{tab:sota}
\centering\small
\begin{tabular}{@{}lllll@{}}
\toprule
Scheme & Type & pk & sig & Extraction model\\
\midrule
CSI-FiSh~\cite{csifish} & signature & large & 263 B & FS, ROM\\
CSI-Otter~\cite{csiotter} & blind sig & 128 B & 8 KB & rewinding, ROM\\
CSI-Otter rGAIP~\cite{csiotter} & blind sig & 512 B & 4 KB & rewinding, ROM\\
Tanuki~\cite{tanuki} & blind sig & varies & varies & ROM, concurrent\\
This work, $k=5$ & PoK/sig & 384 B & 24.1 KiB & online, QROM\\
This work, $k_{\mathrm{eff}}=2$ & PoK/sig & 384 B & 11.6 KiB & online, QROM\\
\bottomrule
\end{tabular}
\end{table}

\begin{figure}[t]
\centering
\includegraphics[width=0.8\textwidth]{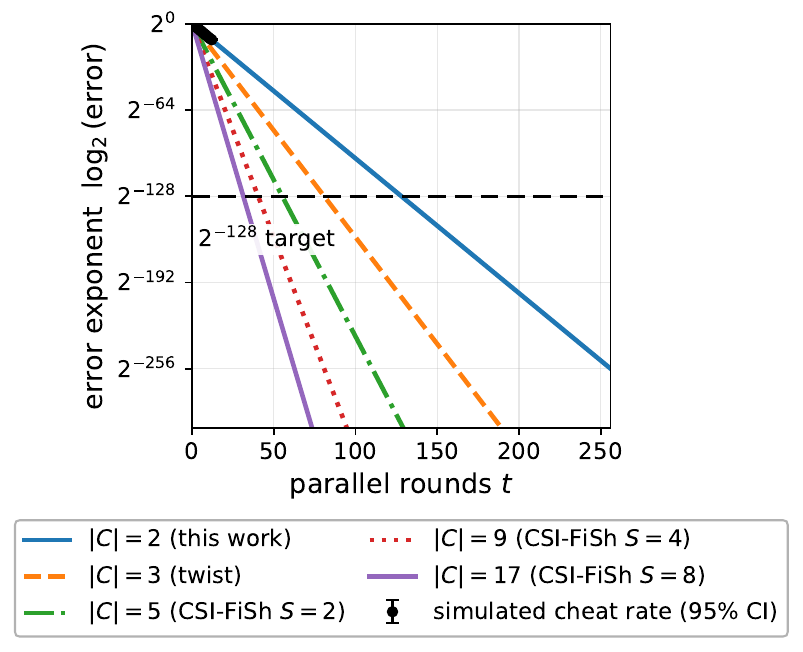}
\caption{Knowledge/soundness error versus parallel rounds $t$ for five
per-round challenge-set sizes: binary (this work and the CSI-FiSh base
protocol), ternary via twists, and CSI-FiSh with $S=2,4,8$ public keys
(challenge set $2S+1$). Provenance: exact formula $(1/|C|)^t$; black points are simulated
cheating success rates ($4\times10^{5}$ runs per point, 95\% confidence
intervals), which match the binary-challenge line at every tested $t$.
All curves need on the order of $t\approx128$ binary rounds to reach the
$2^{-128}$ line; larger challenge sets buy a constant factor at the cost
of larger keys (Prop.~\ref{prop:obstruction}).}
\label{fig:soundness}
\end{figure}

In a classical Schnorr protocol the
verifier checks a large challenge with one exponentiation, because group
elements have composable coordinates. Curves under a group action do not:
the only way to move $c$ steps through the isogeny graph is to take $c$
steps. The verifier is a walker without a map shortcut, so large
challenges are physically unverifiable in one operation, and every
group-action protocol pays for soundness in rounds. This is why CSI-FiSh
enlarges the challenge set only by publishing more keys and CSI-Otter only
by the quadratic twist, each a small constant factor per
round~\cite{csifish,csiotter}. Figure~\ref{fig:soundness} quantifies the
consequence: all curves need on the order of a hundred rounds to reach the
$2^{-128}$ line, and CRT decomposition moves none of them.

\begin{remark}[Leakage of hop curves]\label{rem:leak}
The interior curves $F_1,\dots,F_{k-1}$ are part of the statement. Each
pair $(F_{i-1},F_i)$ is a standalone vectorization instance in the subgroup
$\langle h_i\rangle$. Small components are recoverable outright and must be
treated as public. The exact security consequence is
Prop.~\ref{prop:decomp}.
\end{remark}

\section{Concrete Instantiation on CSIDH-512 and Verified Results}
This section turns the abstract protocol into a fully specified object and
supplies three layers of evidence for it. Result 1 is a machine
verification of the algebraic layer over the exact class number. Result 2
is the security-decomposition bound, the paper's central honest finding
about the cost of the design, together with its deployment consequences.
Result 3 reports protocol-level simulations, with the scope of what
simulation can and cannot establish stated explicitly.

\label{sec:concrete}
The only parameter set satisfying Def.~\ref{def:ks} today is CSIDH-512.
The class group is cyclic, generated by $[\langle 3,\pi-1\rangle]$, of
order~\cite{csifish}
\begin{align*}
N ={}& 3\cdot 37\cdot 1407181\cdot 51593604295295867744293584889\\
     &\cdot\, 31599414504681995853008278745587832204909
     \;\approx\; 2^{257.14},
\end{align*}
a product of $k=5$ distinct primes of bit lengths $2,6,21,96,135$. All
protocol data $u_i,v_i,\varepsilon_i$ are therefore fully determined.

\subsubsection*{Result 1: Machine-verified Algebraic Layer}
Because the action is free and transitive, the curve set is a
$\Zq{N}$-torsor: fixing $E_0$, every curve equals $[g^{a}]\act E_0$ for
exactly one $a\in\Zq{N}$, and acting by $[g^{x}]$ adds $x$ to the exponent.
Theorems~\ref{thm:comp}--\ref{thm:ss} concern only this torsor algebra, so
they can be checked exactly in the exponent model without computing any
isogeny. Our published Python script instantiates key generation, both
branches of~\eqref{eq:verify}, the simulator of Thm.~\ref{thm:zk}, and the
extractor~\eqref{eq:extract} over the exact modulus $N$, for $10^{4}$
independent random instances. All $10^{4}$ pass: the chain~\eqref{eq:chain}
ends at $E_1$; honest transcripts verify under both challenges; simulated
transcripts verify for every challenge; and~\eqref{eq:extract} returns the
planted secret exactly, every time. Figure~\ref{fig:measured} reports the
measured cost of each algebraic operation. We state plainly what this
establishes and what it does not: it is an exact independent confirmation
of the algebra behind Thms.~\ref{thm:comp}--\ref{thm:ke}, and it says
nothing about isogeny running time or hardness, since no isogeny is
computed. \emph{Physical interpretation.} The verification separates the
scheme into two layers that behave like circuit and signal path in a
measurement device. The algebraic layer (exponent bookkeeping, extraction,
simulation) is exact, deterministic, and now independently confirmed; any
failure of a future implementation must therefore originate in the isogeny
layer, which narrows debugging and auditing to one component.
Figure~\ref{fig:measured} additionally shows the algebraic layer costs
microseconds while each isogeny evaluation costs about 40 ms on
CSIDH-512~\cite{csifish}, five orders of magnitude more: in any real
device the isogeny arithmetic is the entire cost, and the extraction
machinery this paper adds is computationally free.

\begin{figure}[t]
\centering
\includegraphics[width=0.66\textwidth]{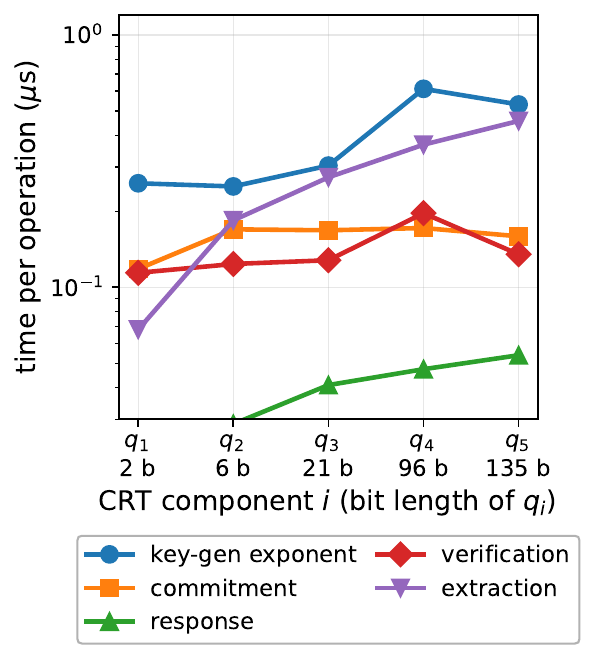}
\caption{Measured cost of the five algebraic operations of
$\Pi_{\mathrm{CRT}}$, per CRT component $q_1,\dots,q_5$ (tick labels give
each component's bit length). Provenance: wall-clock measurement of this
paper's verification script (Python big-integer arithmetic, 3000
repetitions per point) over the exact CSIDH-512 modulus; the algebraic
layer only, excluding isogeny evaluation, which costs about 40 ms per
action~\cite{csifish} and dominates any implementation. The figure shows
that extraction and simulation are sub-microsecond exact formulas, the
claimed qualitative property.}
\label{fig:measured}
\end{figure}

\subsubsection*{Result 2: Security Decomposition}
\begin{proposition}[Cost of publishing hop curves]\label{prop:decomp}
With the hop curves $F_0,\dots,F_k$ public, full key recovery reduces to
$k$ independent vectorization instances in the subgroups
$\langle h_i\rangle$, and a generic meet-in-the-middle adversary recovers
$s$ using $O(\sum_i\sqrt{q_i})=O(\sqrt{q_{\max}})$ action evaluations. For
CSIDH-512 this is about $2^{67.3}$ evaluations, versus about $2^{128.6}$
for the monolithic public key $(E_0,E_1)$.
\end{proposition}
\begin{proof}
Each pair $(F_{i-1},F_i)$ satisfies $F_i=[h_i^{\sigma_i}]\act F_{i-1}$ with
$\sigma_i\in\Zq{q_i}$. Write $\sigma_i=a+b\lceil\sqrt{q_i}\rceil$ with
$0\le a,b<\lceil\sqrt{q_i}\rceil$. The adversary tabulates
$[h_i^{-a}]\act F_i$ for all $a$ and matches against
$[h_i^{\,b\lceil\sqrt{q_i}\rceil}]\act F_{i-1}$ for all $b$; a collision
reveals $\sigma_i$, hence $s_i$, and CRT recombines $s$. The cost per
component is $O(\sqrt{q_i})$ evaluations and the sum is dominated by
$q_{\max}$. The exponents are $\tfrac12\log_2 N=128.57$ and
$\tfrac12\log_2 q_{\max}=67.3$.
\end{proof}

Splitting one lock into $k$ smaller
locks lets a burglar pick each lock separately: total effort is set by the
hardest single lock, not by their product. Publishing the hop curves is
what hands the burglar the separated locks, because each adjacent pair of
curves isolates one component. The monolithic public key keeps the locks
fused; no method is known to separate the components from $(E_0,E_1)$
alone, precisely because a group action admits no homomorphism the
attacker could apply.

The three small factors ($3\cdot37\cdot1407181$, together $27.22$ bits) are
recoverable outright and carry no entropy; effective secret entropy is
$229.9$ bits, and the classical floor is set by Prop.~\ref{prop:decomp} at
about $2^{67}$ evaluations. The CSIDH-512 instantiation therefore does not
reach a 128-bit classical target. The construction should be read as
correct and complete in structure today, and quantitatively secure on
future known-structure class groups whose prime factors are individually
large (at least about 256 bits each). Figure~\ref{fig:surface} shows the
full trade-off surface: the attack exponent as a function of the class
number size and the number of equal-size components. Two mitigations exist
at CSIDH-512: prove only the two large components ($k_{\mathrm{eff}}=2$,
treating the small ones as public), which removes dead weight but not the
$\sqrt{q_{\max}}$ floor; or keep interior hop curves unpublished, which
breaks the verification equation~\eqref{eq:verify} and is an open
protocol-design problem. We record this trade-off rather than hide it: the
CRT structure that enables algebraic extraction is exactly the structure a
decomposition adversary exploits.

The reading of Table~\ref{tab:sota} is direct: prior schemes are smaller,
and what this work adds is a different extraction model, with no rewinding,
QROM validity, and only additive losses (Sec.~\ref{sec:qrom}). Whether the
trade is worthwhile depends on the forking cost at the blind-signature
layer, which we leave open. Code-based blind
signatures face the same rewinding-loss issue~\cite{blazy}, so an
online-extractable group-action protocol is of interest beyond isogenies.
Quantum security of every CSIDH-512 scheme, ours included, is bounded by
Kuperberg-style sieves~\cite{peikert}; Prop.~\ref{prop:decomp} additionally
lowers the classical side for our variant. Both facts argue for larger
future parameter sets.

\begin{figure}[t]
\centering
\includegraphics[width=0.8\textwidth]{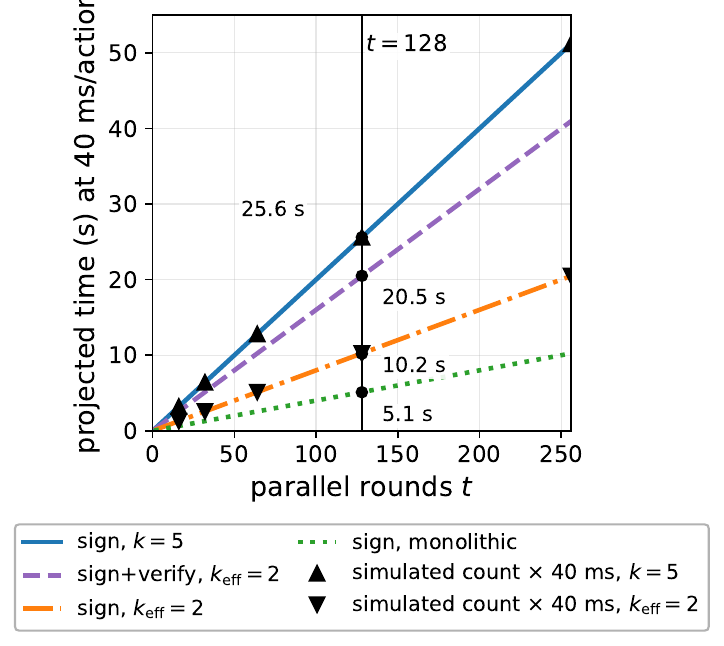}
\caption{Projected time versus rounds $t$, counting one group action per
commitment or verification per hop per round, at the published rate of
40 ms per CSIDH-512 action~\cite{csifish}. Provenance: projection; curve heights are instrumented action counts
from the released simulator (triangular markers) multiplied by the cited
published measurement, not a benchmark of this scheme. Marked values at $t=128$: 25.6 s to sign with
$k=5$; 10.2 s to sign and 20.5 s for a full sign-plus-verify round trip
with $k_{\mathrm{eff}}=2$; 5.1 s to sign monolithically. The CRT variant
costs a factor $k_{\mathrm{eff}}$ in isogeny evaluations, the dominant
cost.}
\label{fig:time}
\end{figure}

\section{QROM Compilation}\label{sec:qrom}
This section explains how the interactive protocol becomes a
non-interactive proof and signature, and why the security argument that
results is quantum-safe in the sense that matters: it never rewinds the
adversary. The content is an application of published theorems, and the
section states exactly which property of our protocol (binary responses
with algebraic extraction) makes the application possible.

Each round has exactly two admissible responses, so Unruh's
transform~\cite{unruh} applies with minimal overhead: for every round the
prover publishes random-oracle commitments to both responses and opens the
one selected by the Fiat-Shamir challenge. Unruh's theorem then gives a
non-interactive zero-knowledge proof of knowledge in the QROM with online
extraction: the extractor reads the adversary's oracle queries and
applies~\eqref{eq:extract}. No rewinding occurs and no forking loss is
incurred.

\begin{proposition}[Instantiation of published theorems]\label{prop:qrom}
Under the vectorization assumption for known-structure instances, the Unruh
compilation of $\Pi_{\mathrm{CRT}}$ with $t$ rounds per hop is a
simulation-sound NIZK proof of knowledge in the QROM with online extraction
and soundness error $2^{-t}$ plus negligible terms; the derived signature
is sEUF-CMA in the QROM with a bound whose dominant loss is the additive
online-extraction terms of~\cite{unruh}, not a multiplicative forking loss.
The Fiat-Shamir analysis of~\cite{dfms} covers the plain compilation.
\end{proposition}

A rewinding proof runs the adversary,
saves its state, and reruns it from the save point, which is impossible to
justify against a quantum adversary whose state cannot be cloned. Online
extraction replaces the rerun with a wiretap: the reduction observes the
adversary's random-oracle queries as they happen and reads the secret out
of them using~\eqref{eq:extract}. The adversary is executed once, in
real time, exactly as in a real attack, which is why the resulting bound
has no multiplicative loss. This proposition applies published theorems
to our protocol; the application is enabled by Thms.~\ref{thm:ss}
and~\ref{thm:ke}, and we claim no new QROM technique.

\section{Conclusion}\label{sec:conc}

On a class group of known structure, CRT decomposition yields a
$\Sigma$-protocol whose extraction is a closed-form ring computation, and
Unruh's transform lifts it to a QROM proof of knowledge with online
extraction; on CSIDH-512 the scheme is fully determined and its algebraic
layer is machine-verified. The same decomposition cannot shorten the
protocol, increases cost by known factors, and caps classical security at
$\sqrt{q_{\max}}$ evaluations, so deployment awaits known-structure
class groups with only large prime factors. Building a blind signature on
this proof of knowledge, and reducing its unforgeability tightly to the
endomorphism ring problem, is left to future work.

\bibliography{refs}

\end{document}